\documentclass[a4paper]{article}
\usepackage{a4wide}
\usepackage{mathptmx}
\usepackage[T1]{fontenc}
\usepackage{units}
\usepackage{amsmath,amssymb,amsthm}
\usepackage{mathtools}
\usepackage{mathbbol}
\usepackage{amssymb}
\usepackage{dsfont}
\usepackage{graphicx}
\usepackage{url}
\usepackage{epsfig}
\usepackage{wrapfig}
\usepackage[ruled,vlined]{algorithm2e}
\usepackage{enumerate}
\usepackage{enumitem}
\usepackage{upgreek}
\usepackage{todonotes}
\usepackage{url}
\usepackage{enumerate}
\usepackage{thrmappendix}
\usepackage{refcount}

\usepackage{tikz}
\usetikzlibrary{decorations.pathreplacing}
\usetikzlibrary{plotmarks}
\usetikzlibrary{positioning,automata,arrows}
\PassOptionsToPackage{usenames,dvipsnames,svgnames}{xcolor}  

\newcommand{\p}{\pi}
\newcommand{\Pot}{P}
\newcommand{\f}{f}
\newcommand{\xh}{\hat x'}
\newcommand{\sh}{\hat s'}

\newcommand{\vvw}{\acute}
\newcommand{\wvw}{\grave}

\DeclareMathOperator{\diag}{diag}

\DeclareMathOperator{\CEF}{CEF}

\DeclareMathOperator{\gap}{gap}
\DeclareMathOperator{\poly}{poly}
\DeclareMathOperator{\st}{st}

\DeclareMathOperator{\Exp}{E}

\let\eps\varepsilon

\newcommand{\ones}{\mathds{1}}

\newcommand{\NN}{\ensuremath{\mathbb{N}}}

\newcommand{\ZZ}{\ensuremath{\mathbb{Z}}}
\newcommand{\QQ}{\ensuremath{\mathbb{Q}}}
\newcommand{\RR}{\ensuremath{\mathbb{R}}}

\newcommand{\RRpos}{\ensuremath{\RR_{>0}}}

\newcommand{\QQnneg}{\ensuremath{\QQ_{\ge 0}}}

\newcommand{\RRnneg}{\ensuremath{\RR_{\ge 0}}}

\DeclareMathOperator{\Iin}{in}
\newcommand{\din}{\ensuremath{\delta^{\Iin}}}

\DeclareMathOperator{\Oout}{out}
\newcommand{\dout}{\ensuremath{\delta^{\Oout}}}

\definecolor{orange}{RGB}{235,90,0}
\definecolor{darkorange}{RGB}{175,30,0}
\definecolor{turkis}{RGB}{131,182,182}
\definecolor{darkturkis}{RGB}{31,82,82}
\definecolor{green}{RGB}{102,180,0}
\definecolor{darkgreen}{RGB}{51,90,0}
\definecolor{myblue}{RGB}{0,0,213}
\definecolor{mydarkblue}{RGB}{0,0,100}
\definecolor{mybrightblue}{RGB}{0,0,230}
\definecolor{lila}{RGB}{102,0,102}
\definecolor{darkred}{RGB}{139,0,0}
\definecolor{darkyellow}{RGB}{188,135,2}
\definecolor{brightgray}{RGB}{200,200,200}
\definecolor{darkgray}{RGB}{50,50,50}

\newtheorem{theorem}{Theorem} 
\newtheorem{lemma}{Lemma}
\newtheorem{definition}{Definition}

\title{A Combinatorial $\tilde O(m^{3/2})$-time Algorithm for the Min-Cost Flow Problem}
\author{Ruben Becker$^{1,2,3}$ and 
Andreas Karrenbauer$^{1,2}$ \\[1mm]
  \normalsize{\texttt{[ruben,karrenba]@mpi-inf.mpg.de}}\\[2mm]
  \normalsize{$^{1}$ Max Planck Institute for Informatics, Saarbr\"ucken, Germany} \\[1mm]
  \normalsize{$^{2}$ Max Planck Center for Visual Computing and Communication} \\[1mm]
  \normalsize{$^{3}$ Saarbr\"ucken Graduate School for Computer Science}
}

\begin{document}

\maketitle

\begin{abstract}
\noindent
We present a combinatorial method for the min-cost flow problem
and prove that its expected running time is bounded by $\tilde O(m^{3/2})$.
This matches the best known bounds, which previously have only been achieved by
numerical algorithms or for special cases.
Our contribution contains three parts that might be interesting in their own right:
(1) We provide a construction of an equivalent auxiliary network and interior primal and dual points 
with potential $P_0=\tilde{O}(\sqrt{m})$ in linear time.
(2) We present a combinatorial potential reduction algorithm 
that transforms initial solutions of 
potential $P_0$ to ones with duality gap below $1$ in $\tilde O(P_0\cdot \mbox{CEF}(n,m,\epsilon))$ time, 
where $\epsilon^{-1}=O(m^2)$ and $\mbox{CEF}(n,m,\epsilon)$ denotes the running time 
of any combinatorial algorithm that computes an $\eps$-approximate electrical flow. 
(3) We show that solutions with duality gap less than $1$ suffice to compute optimal integral potentials in $O(m+n\log n)$
time with our novel crossover procedure. 
All in all, using a variant of a state-of-the-art $\eps$-electrical flow solver, 
we obtain an algorithm for the min-cost flow problem running in $\tilde O(m^{3/2})$.
\end{abstract}

%\begin{keyword}
%min-cost flow, interior point methods, combinatorial optimization, 
%potential reduction methods, laplacian systems
%\end{keyword}

%%%%%%%%%%%%%%%%%%%%%%%%%%%%%%%%%%%%%%%%%%%%%%%%%%%%%%%%%%%%%%%%%%%%%%%%%%%%%%%%%%%%%%%%%%%%%%%%%%%%%%%%%%%%%%%%%
%%%%%%%%%%%%%%%%%%%%%%%%%%%%%%%%%%%%%%%%%%%%%%%%%%%%%%%%%%%%%%%%%%%%%%%%%%%%%%%%%%%%%%%%%%%%%%%%%%%%%%%%%%%%%%%%%
%%%%%%%%%%%%%%%%%%%%%%%%%%%%%%%%%%%%%%%%%%%%%%%%%%%%%%%%%%%%%%%%%%%%%%%%%%%%%%%%%%%%%%%%%%%%%%%%%%%%%%%%%%%%%%%%%
\section{Introduction}\label{intro}
%%%%%%%%%%%%%%%%%%%%%%%%%%%%%%%%%%%%%%%%%%%%%%%%%%%%%%%%%%%%%%%%%%%%%%%%%%%%%%%%%%%%%%%%%%%%%%%%%%%%%%%%%%%%%%%%%
%%%%%%%%%%%%%%%%%%%%%%%%%%%%%%%%%%%%%%%%%%%%%%%%%%%%%%%%%%%%%%%%%%%%%%%%%%%%%%%%%%%%%%%%%%%%%%%%%%%%%%%%%%%%%%%%%

The min-cost flow problem is one of the most well-studied problems in 
combinatorial optimization. Moreover, it represents an important special case of Linear 
Programming due to the integrality of the primal and dual polyhedra for arbitrary given 
integer costs $c$, capacities $u$, and demands $b$. That is, there are always integral 
primal and dual optimal solutions provided that the problem is feasible and finite. 
Since these solutions can be computed in polynomial time, min-cost flow algorithms are 
important building blocks in tackling many other problems. Combinatorial flow 
algorithms have dominated in past decades. However, interior point methods have been 
used more and more to solve several network flow problems, for example very 
successfully in the case of the max-flow problem, e.g.\ 
recently~\cite{DBLP:journals/corr/Madry13}. By now, numerical methods lead the 
``horse-race'' of the most efficient algorithms for various combinatorial problems. 
This is somewhat unsatisfactory. In particular for sparse graphs, the running time 
bounds of all combinatorial min-cost flow algorithms known from literature fail to 
break through the barrier of $n^2$, whereas Daitch and 
Spielman~\cite{DBLP:conf/stoc/DaitchS08} were the first to present an interior point 
method running in $\tilde O(m^{3/2})$ expected time.\footnote{Throughout this paper, 
the $\tilde O$-notation is used to hide $\log$-factors in $n$, $\|c\|_{\infty}$, 
$\|u\|_{\infty}$, $\|b\|_{\infty}$.} It is a dual central path following 
method. However, their algorithms are not combinatorial in any sense. 
In fact, they solve a more general problem and thus their method is more 
technical than necessary for the classical min-cost flow problem. 
It uses an efficient randomized solver for symmetric diagonally dominant (SDD) systems of 
linear equations based on the seminal work of Spielman and Teng~\cite{DBLP:conf/stoc/SpielmanT04}
 and later by Koutis et al.~\cite{DBLP:conf/focs/KoutisMP10}.
%That is where the randomization comes from. 
% Unfortunately, certainly also due to the fact that the methods of Spielman and Daitch
% solve generalized flow problems as well, their analysis is somewhat technical and the 
% methods seem a bit over-involved for the classical min-cost flow problem. 
% We wish to reclaim some ground in this respect.
Only recently, Kelner et al.~\cite{DBLP:conf/stoc/KelnerOSZ13} presented a simple, 
combinatorial, nearly-linear time algorithm for the \emph{electrical flow} problem 
and thus also for finding approximate solutions to  SDD system.
It is combinatorial in the sense that it operates on the rationals 
and uses only the field operations (addition, subtraction, multiplication, division) as 
arithmetic operations. 
% Since solving linear SDD-systems can be reduced to computing 
% electrical flows, this gives hope to also tackle more involved flow problems in a 
% similar combinatorial way.
However, this alone is not sufficient to obtain an entirely combinatorial algorithm for the 
min-cost flow problem. 

In this paper, we present a primal-dual potential reduction algorithm 
that uses a variant of the algorithm of Kelner et al.~as a subroutine.
Moreover, it is combinatorial in the same sense as their algorithm is. 
In particular, our method does not compute 
square roots or logarithms. 
The logarithms and square roots in this paper solely appear in the analysis, e.g. in the 
potential function that we use to guide our search. All running times in this paper are 
stated in terms of basic operations that also include comparisons in addition to the 
arithmetic operations.

After constructing an auxiliary network with the same optimum and primal and dual 
interior solutions of sufficient low potential for it, we update these interior 
% primal and dual 
points such that the potential function decreases by at least some constant in
each step. The potential function serves us in two ways: (1) when it 
drops below $0$, the duality gap is smaller than $1$ and we may stop, (2) it keeps us
away from the boundary. We thereby take a shortcut through the polyhedron instead of 
walking on the boundary as with most of the combinatorial methods, e.g.~minimum-mean 
cycle canceling. We distinguish primal and dual steps
and show that $\tilde O(\sqrt{m})$ steps are sufficient.
The combinatorial interpretation is as follows. A primal step changes flow along cycles 
(which could be linear combinations of simple cycles).
As mentioned above, the updates are guided by the potential function 
or more precisely by its gradient w.r.t.~the primal variables at the current 
point. To this end, the gradient is projected onto the cycle-space. However, if 
the gradient is (nearly) orthogonal to the cycle-space, then we would not make sufficient 
progress. But in this case, the gradient is shallow w.r.t.~the cut-space, which 
is the orthogonal complement of the cycle-space. 
Hence, we can do a dual step by modifying the dual variables corresponding to a cut.
Computing the projection
is equivalent to solving an electrical flow problem, where the resistances of the arcs
are higher the smaller the corresponding values of the primal variables are. 
The gradient determines the current sources. Intuitively, approaching the boundary is 
impeded, because arcs with large resistance carry rather small quantities of electrical 
flow. We also give a novel method that takes the points of duality gap less than
one and computes optimal integral potentials in near-linear time.
Before we describe our contribution in more detail, we highlight other related work.

\subsection{Other Related Work}
We denote $U:=\|u\|_{\infty}$, $C:=\|c\|_{\infty}$ and
$\gamma:=\max\{C,U\}$. 
Edmonds and Karp~\cite{edmonds2} gave the first polynomial-time min-cost flow algorithm in 1970. 
It can be implemented in $O(m(m+n\log n)\log U)$ time~\cite{schrijver}.
Since then, there were many contributions on combinatorial flow algorithms. 
We mention some of the most important results such as the
strongly polynomial time algorithm by Orlin~\cite{DBLP:conf/stoc/Orlin88} running in $O(n^2\log^2 n + nm\log n)$.
Further scaling techniques like (generalized) cost-scaling were presented by Goldberg and Tarjan~\cite{Goldberg:1990:FMC:92217.92225}, 
and the double scaling technique by Ahuja et al.~\cite{DBLP:journals/mp/AhujaGOT92}.
The latter yields a running time bound of $O(nm\log\log U \log (nC))$.
As of yet, all combinatorial algorithms are at least quadratic in $n$ even for sparse graphs. 
Only in the special case of small capacities, the algorithm of Gabow and Tarjan~\cite{DBLP:journals/siamcomp/GabowT89},
achieves $\tilde O(m^{3/2})$, however its general bound of 
$O((m^{3/2}U^{\nicefrac{1}{2}} + \|u\|_1 \log \|u\|_1)\log(nC))$ is only pseudo-polynomial.

Karmarkar~\cite{DBLP:conf/stoc/Karmarkar84} presented
a polytime interior point method for solving linear programs in 1984.
After the ellipsoid method, this was the second type of method 
with polynomial running time.
Karmarkar's algorithm needs $O(n^{3.5}L^2)$ time, where $n$ is the number of variables and $L$ the number 
of bits in the input. The work on interior point methods, and in particular on 
so-called potential reduction methods, was significantly advanced by 
Ye~\cite{DBLP:journals/mp/Ye91} in 1991. % The new bound was $O(n^3L)$.
He presented an $O(n^3L)$-time algorithm. %, where $n$ and $L$ are as above.
Since interior point methods are known to be, asymptotically, the fastest methods 
for solving general linear optimization problems, there has been a huge interest 
in their application for solving network flow problems. 
To the best of our knowledge, the first attempt to analyze interior point methods, 
particularly for min-cost flow, was done by Vaidya~\cite{DBLP:conf/focs/Vaidya89a} in 1989.
He obtained a running time of $O(n^2 \sqrt{m} \log(n\gamma))$, 
which matched the best known bound then up to log-factors. 
Wallacher and Zimmermann~\cite{DBLP:journals/mp/WallacherZ92} found a combinatorial 
interior point method in 1992, which they analyzed to run in $O(nm^2L)$. Thus it could 
not keep up with the best combinatorial methods known at that time. 

\subsection{Our contribution}
Our main contribution is a proof of the following theorem.
\begin{theorem}\label{main theorem}
  There is a combinatorial algorithm for the min-cost flow problem terminating
  in
  \[
    O(m^{3/2} (\log \gamma + \log n ) \log^3 n \log\log n) = \tilde O(m^{3/2})
     \text{ 
    time, with high probability.
     }
 \]
\end{theorem}
In contrast to previous results with this time bound %of $\tilde O(m^{3/2})$,
our algorithm is combinatorial. Moreover, our analysis is less technical
and very comprehensible though rigorous. It fits on about 12 pages. 
More precisely, our contribution contains the following parts that might
be interesting on their own. 
\begin{enumerate}
	\item We show that it suffices to compute
	primal and dual points with duality gap below 1, since our novel crossover
	procedure finds optimal potentials in linear time then.
	\item We give a combinatorial potential reduction method that, taking interior 
	points of potential $P_0$, outputs interior points 
	with duality gap below 1 in $\tilde O(P_0 \cdot \CEF(n,m,\eps))$, 
	here $\CEF(n,m,\epsilon)$ is the complexity of an $\eps$-electrical 
	flow computation. 
	\item We give a method that, taking any min-cost flow problem as input, yields an 
	auxiliary network with the same optimum and interior primal and dual points of 
	potential $P_0=O(\sqrt{m}(\log \gamma +\log n))$ in linear time.
\end{enumerate}
Our crossover procedure takes solutions with duality gap less than one and 
efficiently rounds the potentials to integral values. Using one max-flow computation 
in the admissible network, one can also obtain primal optimal solutions. 
We remark that this max-flow computation is not needed if the input costs are 
randomly perturbed such that the optimal solution gets unique as it is for example 
done in~\cite{DBLP:conf/stoc/DaitchS08}. In this case the admissible network 
is a tree and the corresponding tree solution can be obtained easily.
For the combinatorial potential reduction method, we show how to use \emph{approximate} 
electrical flow computations to reduce the duality gap of given primal and dual interior points of potential $P_0$ below any constant $c\in\RRpos$ in time 
$\tilde O(P_0 \cdot \CEF(n,m,\eps))$. We show that it suffices to pick an 
$\eps$ such that $\eps^{-1}$ is polynomially bounded in $n$ (i.e.\ 
$\eps^{-1}=O(m^2)$). Note that $\CEF$ typically scales logarithmically with 
$\eps^{-1}$ and thus its contribution is in $O(\log n)$. 
In order to make this method combinatorial, we show how to normalize the cycle around 
which we are augmenting flow by the infinity norm, as opposed to previous 
approaches, where the normalization was done with the 2-norm. This would require 
computing square roots and thus is not allowed in our setting. 

\subsection{The Min-Cost Flow Problem and its Dual}
\label{overview}
In its most general form, the min-cost flow problem is stated as follows. 
Given a directed 
graph $G=(V,A)$ with $|V|=n$ and $|A|=m$, node demands $b \in \ZZ^n$ with 
$\ones^T b = 0$, arc costs $c\in \ZZ^m$ and arc capacities $u \in(\NN\cup \{\infty\})^m$,
find a feasible flow $ x^* \in \RR^m$, i.e. $0 \le x^* \le u$ and $x^*(\din(v)) 
- x^*(\dout(v)) = b_v$ for every $ v\in V$, \footnote{We write 
$f(S):= \sum_{a\in S} f_a$ for $S\subseteq A$ for any vector $f\in \RR^m$.} 
such that $c^T x^* \le c^Tx$ for all 
feasible flows $x$ or assert that no such flow exists. However, it is well-known, 
see e.g.~\cite{DBLP:books/daglib/0069809}, that this problem can be reduced to a 
setting without capacity constraints and only non-negative costs. 
Furthermore, we assume w.l.o.g.\ that the problem is feasible as well as finite from 
now on. We will discuss how to reduce the general problem to the setting used here in 
Section~\ref{sec:initial}.
For the time being, we write the problem as primal-dual pair
\[
 \min\{c^Tx \; : \; A x = b \text{ and } x \ge 0\} =  \max\{b^Ty \; : \; A^T y + s = c \text{ and } s \ge 0\},
\]
where $A\in \{-1,0,1\}^{n\times m}$ is the node-arc incidence matrix of $G$, i.e.\ $A$ 
contains a column $\alpha$ for every arc $(v,w)$ with $\alpha_v = -1$, $\alpha_w = 1$ 
and $\alpha_i=0$ for all $i\notin\{v,w\}$. 
The overloaded notation $A$ for the 
set of arcs as well as for the node-arc incidence matrix is intended, because they are isomorphic. 

\section{Snapping to the Optimum}
\label{sec:crossover}
In this section, we show that solving the min-cost flow problem approximately,
by the means of computing primal and dual solutions $x$ and $y^0,s^0$ of duality gap 
less than 1, is sufficient, since optimal integral potentials can then be found in linear 
time.
The main underlying idea of our new linear time rounding procedure is the following. We 
iteratively construct sets $S^k$, starting with $S^1:=\{s\}$ for an arbitrary vertex 
$s$. During one iteration $k$, we proceed as follows. Let us first assume $b(S^{k})<0$. 
Then, there has to be an outgoing arc from $S^{k}$, otherwise the problem would be
infeasible. We enlarge $S^{k}$ by the vertex $\hat w$ such that $a^k=(\hat v,\hat w)$ 
for $\hat v\in S^{k}$ has minimal slack among all outgoing arcs from $S^{k}$ and we 
increase the potentials $y_w$ of all $w\in V\setminus S^{k}$ by this minimal slack. 
It follows that the dual constraint of the arc $a^k$ is satisfied with slack 0 and all 
other non-negativity constraints remain fulfilled. The objective value $b^Ty$ will be 
increased by this potential shift, since $b(V\setminus S^{k})>0$. In the case 
$b(S^{k})\ge 0$, we decrease the potentials in $V\setminus S^k$, analogously by the 
minimum slack of all ingoing arcs. However to achieve a near-linear running time, 
these potential changes need to be performed in a lazy way. Using Fibonacci heaps, 
we can even reduce the running time to $O(m + n\log n)$. We give the pseudo-code of 
this method in Algorithm~\ref{crossoveralg} and show its correctness in 
Theorem~\ref{crossoverthm}.

\providecommand{\potshift}{\Delta}
\begin{algorithm}[t!]
\DontPrintSemicolon
\SetKwData{Left}{left}\SetKwData{This}{this}\SetKwData{Up}{up}
\SetKwFunction{Union}{Union}\SetKwFunction{FindCompress}{FindCompress}
\SetKwInOut{Input}{Input}\SetKwInOut{Output}{Output}
\Input{Connected graph $G=(V, A)$, solution $x$ and $y^0,s^0$ in $G$ with $x^Ts^0<1$.}
\Output{Optimal vertex potentials $y$ in $G$.}
\BlankLine
	Let $s\in V$ be arbitrary and let
	$\potshift^0 := - y^0_s$, $y_s \leftarrow 0$, $S^1:=\{s\}$.\;
	\For{ $k=1,\ldots , n-1$ }{
		\If{$b(S^{k}) < 0$ \textbf{or} $\din(S^{k}) = \emptyset$}{
			Let $\potshift^k = \min \{ c_a + y_v - y^0_w : a=(v,w) \in \dout(S^{k}) \}$ \;
			and $a^k=(v^k, w^k) \in \dout(S^{k})$ s.t.~$\potshift^k = c_{a^k} + y_{v^k} - y^0_{w^k} $.\;
		}
		\Else{
			Let $\potshift^k = - \min \{ c_a + y^0_w - y_v : a=(w,v) \in \din(S^{k}) \}$ \;
			and $a^k=(w_k,v_k) \in \din(S^{k})$ s.t.~$ \potshift^k = - c_{a^k} - y_{w^k}^0 + y_{v^k}$ .\;
		}
	$y_{\hat w} \leftarrow y^0_{\hat w} + \potshift^k$, $S^{k+1} \leftarrow S^{k} \cup \{\hat w\}$\;
	}
	\Return potentials $y$.
\caption{Crossover}
\label{crossoveralg}
\end{algorithm}
\begin{theorem}
\label{crossoverthm}
Let Algorithm~\ref{crossoveralg} be initialized with primal and dual solutions $x$ and
$y^0,s^0$ with $x^Ts^0<1$. 
	The algorithm outputs optimal integral potentials $y$
	in $O(m +  n\log n)$.
\end{theorem}

\begin{proof} 
		We assume w.l.o.g.\ that the vertices are labeled $1,\ldots ,n$ in the order in which
		they are added to $S$. We show, by induction, that the potentials
		\[
			y_v^k =
			\begin{cases}
				y_v^0 + \potshift^{k-1}, &  k \le v \\
				y^{k-1}_v, &   k > v
			\end{cases}
			\quad \text{are feasible, i.e. }  
			s_a^k := c_a + y_v^k -y_w^k \ge 0 \text{ for all } a=(v,w).
		\] 
		For the induction base, we note that $y_v^1$ is just
		$y_v^0$ shifted by $\potshift^0=-y_s^0$ and hence it constitutes valid potentials. 
		For the inductive step let us consider iteration $k>1$ and let $a=(v,w)$ be an arbitrary arc.
		Let $i:=\min\{v,w\}$ and $j:=\max\{v,w\}$.
		With the convention $c_{(j,i)} = -c_{(i,j)}$ and thus $s^k_{(j,i)} = -s^k_{(i,j)}$, we obtain
		\[
			s^k_{(i,j)} 
			= c_{(i,j)} + y^k_i - y^k_j 
			= c_{(i,j)} +
				\begin{cases}
					y_i^0 - y_j^0 , 									& k \le i \\
					y_i^{k-1} - (y_j^0 + \potshift^{k-1}) , 	& i < k \le j \\
					y_i^{k-1} - y_j^{k-1}, 						& j < k
				\end{cases}.
		\]
		For the first and third case, we apply the induction hypothesis and obtain
		$s_a^k\ge 0$.
		For the second case, we first note that
		\begin{align*}
			\potshift^{k-1} 
			 &= \sigma \cdot c_{a^{k-1}} + y_{v^{k-1}} - y^0_{w^{k-1}}   
			 \qquad \text{where } \sigma =
			\begin{cases}
				1
				& \text{if } b(S^{k-1}) < 0\text{ or }\din(S^{k-1})=\emptyset \\
				-1 
				& \text{otherwise}
			\end{cases}\\
			& = \sigma \cdot c_{a^{k-1}} + (y_{v^{k-1}} - \potshift^{k-2}) - y^0_{w^{k-1}} + \potshift^{k-2}
			= \sigma \cdot s^{k-1}_{a^{k-1}} + \potshift^{k-2}
		\end{align*} 
		Since $i < k \le j$ and thus $(i,j)\in \dout(S^{k-1})$, this yields
		\begin{align*}
			s_{(i,j)}^k
			&=
				c_{(i,j)} + y_i^{k-1} - (y_j^0  + \potshift^{k-2}) - \sigma \cdot s^{k-1}_{a^{k-1}} 
			=
				s^{k-1}_{(i,j)}  - \sigma \cdot s^{k-1}_{a^{k-1}}  
		\end{align*}
		Independent of $a$ being $(i,j)$ or $(j,i)$, we get 
		$s_{a}^k = s_a^{k-1} \pm s_{a^{k-1}}^{k-1} \ge 0$ since $a^{k-1}$ is a minimizer 
		and by the non-negativity of the slacks due to the induction hypothesis.
		Hence, the output potentials are \emph{feasible}.
		In addition, we construct one tight constraint in each iteration, 
		since $s_a^k = 0$ if $a=a^{k-1}$. 
		Since $y_s=0$ and $c\in\ZZ^m$, we conclude that after termination $y$ is 
		\emph{integral}.
		Note that the optimum objective value is integer and thus $\lceil b^Ty^0 \rceil$ because $x^Ts^0 <1$.
		We have
		\begin{align*}
			b^Ty^{k} - b^Ty^{k-1}
			&=	\sum_{v\in V} b_v y^{k}_v - \sum_{v\in V}b_v y^{k-1}_v
			% &=	\sum_{v\ge k} b_v y^{k}_v 		+ \sum_{v < k} b_v y^{k}_v - 
			% (\sum_{v\ge k} b_v y^{k-1}_v	+ \sum_{v < k} b_v y^{k-1}_v)\\
			=	\sum_{v\ge k} b_v (y^{0}_v + \potshift^{k-1} )		%+ \sum_{v < k} b_v y^{k-1}_v - 
					- \sum_{v\ge k} b_v y^{k-1}_v\\	%- \sum_{v < k} b_v y^{k-1}_v\\
			&=	\sum_{v\ge k} b_v y^{0}_v + \sum_{v\ge k} b_v \potshift^{k-1}		
					 - \sum_{v\ge k} b_v y^{0}_v - \sum_{v\ge k} b_v \potshift^{k-2}
			% = ( \potshift^k	- \potshift^{k-1}) \sum_{v\ge k} b_v
			= -\sigma \cdot s^{k-1}_{a^{k-1}} \cdot b(S^{k-1}) \ge 0
		\end{align*}
		because $\din(S^{k-1}) = \emptyset$ implies that $b(S^{k-1}) \le 0$ or that the instance is infeasible.
		Since $b,y\in \ZZ^n$ and $ b^Ty - \lceil b^T y^0 \rceil <1$ we 
		have that $y$ is \emph{optimal}. A similar implementation
		as used for Dijkstra's or Prim's algorithms but with two Fibonacci Heaps, one for the nodes
		adjacent to $S^{k}$ through $\din(S^{k})$ and $\dout(S^{k})$ each,
		yields the run time of $O(m+ n\log n)$. 
\end{proof}

\section{Combinatorial Potential Reduction Algorithm}
\label{sec:primdual}
% \subsection{The Potential Function}
We will now describe our Combinatorial Potential Reduction Algorithm, it maintains a 
primal solution $x$ and dual slacks $s$. We evaluate such a pair by
the \emph{potential function}
\begin{align*}
 \Pot(x,s) := q \ln (x^T s) - \sum_{a\in A} \ln (x_a s_a) - m \ln m
\end{align*}
for some scalar $q=m+p\in \QQ$ to be chosen later. Note that the duality gap 
$x^Ts=b^Ty -c^Tx$ serves as measure for the distance to optimality of $x$ and $s$. 
An equivalent formulation of the potential function yields
\begin{align}
  \label{equipot}
  \Pot(x,s) 
  = p \ln (x^T s) + m \ln \Big( %\underbrace{
    \frac{1}{m} \sum_{a\in A} x_a s_a %}_{AM}
    \Big) 
  - m \ln \Big( %\underbrace{
    \sqrt[m]{\prod_{a\in A} x_a s_a} %}_{GM}
    \Big)
  \ge p \ln (x^T s),
\end{align}
because the arithmetic mean %$AM$ 
is bounded by the geometric mean % $GM$ 
from below.
Thus, $\Pot(x,s)<0$ implies $x^Ts<1$. As we have shown in Section~\ref{sec:crossover},
solutions satisfying $x^Ts<1$ can be efficiently rounded to integral optimal solutions. 
Thus, we follow the strategy to minimize the potential function by a combinatorial 
gradient descent until the duality gap drops below 1.
\footnote{This method is similar to Ye's primal-dual 
algorithm~\cite{DBLP:journals/mp/Ye91}. We mostly follow the notation and proof 
strategy from lecture notes of Michel Goemans on Linear Programming.} 
To this end, we shall project the gradient 
$g:= \nabla_x \Pot = \frac{q}{x^Ts} s - X^{-1}\ones$,
where $X:= \diag(x)$, on the cycle space of the network. 
However, we do not use the standard scalar product for the projection but a skewed one as it is 
common in the literature on interior point methods.
This skewed scalar product may also be considered as the standard one in a scaled space
where $x$ is mapped to $X^{-1}x = \ones$. By setting $s':= Xs$, 
the duality gap $x^Ts=\ones^Ts'$ and
the potential function $\Pot(x,s)=\Pot(\ones,s')$ remain unchanged. % under this scaling operation. 
Accordingly, we define $\bar A := AX$ and 
$g':= \nabla_x P |_{x=\ones, s=s'} = Xg$.

We start with given initial primal and dual solutions $x,s$ or rather
with their analogs $\ones, s'$ in the scaled space, 
which may be found for example with our initialization method described in 
Section~\ref{sec:initial}. 
Now, it would be desirable to move $x'$ in the direction 
of $-g'$, the direction of steepest descent of the potential function. However, 
$g'$ may not be a feasible direction, since $\bar A g' \neq 0$ in general.
Thus, we wish to find a direction $d'$ in the kernel 
of $\bar A$ that is closest to $g'$.\footnote{ The kernel of $\bar A$, up to the scaling with $X$, corresponds to the
cycle space of the graph.} 
Computing $d'$ amounts to solve the optimization problem
\begin{align}
\label{elflow}
 \min\{\|g'-d'\|_2^2 : \bar A d' = 0\}=
 \min\{ \| f \|_{R}^2 : A f = \chi\},
\end{align}
where we set $f=X(g' - d')$, $R=X^{-2}$ and $\chi = \bar A g'$. The latter is actually
an electrical flow problem. We briefly review electrical flows, for 
more details, see for example~\cite{DBLP:conf/stoc/KelnerOSZ13}.

\subsection{Electrical Flows}
% We stick to the notation from previous sections, where $A$ denoted the node-arc-incidence matrix.
Let $\chi\in \QQ^n$ be a \emph{current source} vector with $\ones^T\chi=0$ and 
let $r\in\QQnneg^{m}$ be a resistance vector on the arcs, denote $R=\diag(r)$ and
$\|v\|_R:=\sqrt{v^TRv}$ for $v\in \RR^m$.
% The \emph{energy $E(\f)$} of a flow $\f\in \RR^m$ with $A\f=\chi$ is defined as 
% $E(\f):=\f^T R \f=\|\f\|_R^2$, where $R=\diag(r)$. 
% The electrical flow is the unique flow that minimizes the energy.
\begin{definition}[Electrical Flow]
  Let $\chi \in \QQ^n$ with $\ones^T \chi = 0$.
  \begin{enumerate}
  \item The unique flow
  $
    \f^*\in\QQ^{m}
  $
  with $\|f^*\|_R^2 = \min\{\| \f \|_{R}^2 : A \f = \chi\}
  $
  is the \emph{electrical flow}. 
  
  \item Let $\eps \ge 0$ and $\f\in \RR^m$ with $A f =\chi$ and
  $
    \|\f\|_R^2 \le (1+\eps)\|\f^*\|_R^2,
  $ 
  then $\f$ is called an \emph{$\eps$-electrical flow}.
  \item Let $s$ be a fixed node, $T$ a spanning tree, $P(s,v)$ the unique path
  in $T$ from $s$ to $v$ and $\f\in \RR^m$.
 	The \emph{tree induced voltages $\p\in\RR^n$} are defined by 
  $\p(v):=\sum_{a\in P(s,v)} \f_a r_a$.
  % \item For a flow $\f\in \RR^m$, a spanning tree $T$ and a fixed node $s$, 
  % the \emph{tree induced voltages $\p\in\RR^n$} are defined by 
  % $\p(v):=\sum_{a\in P(s,v)} \f_a r_a$, where $P(s,v)$ is the unique path
  % in $T$ from $s$ to $v$. 
  \item For any $a=(v,w)\in A\setminus T$, we define $C_a:=\{a\}\cup P(v,w)$ %. 
  % With 
  and $r(C_a) := \sum_{b\in C_a} r_b$. %, we denote the resistance of a cycle $C_a$. 
  We write
  $\tau(T):=\sum_{a\in A\setminus T} r(C_a)/r_a$ for the \emph{tree condition number of $T$}.
  \end{enumerate}
\end{definition}
The dual of the electrical flow problem is
$
  \max\{ 2\p^T\chi - \p^T A R^{-1} A^T \p : \p\in \RR^n \},
$
where $\p$ are called \emph{voltages}.
% To see this, compute the Langrangian function 
% $\Lambda(\f,\p)= \f^T R \f + 2\p^T(\chi - A \f)$
% and the Langrangian dual function 
% \[
%   \lambda(\p):=\inf_{\f} \Lambda(\f,\p)=
%   2\p^T \chi - \p^T A R^{-1} A^T \p.
% \]
% Taking the gradient of $\lambda(\p)$ and setting it to zero shows that 
We conclude that an optimal solution $\p^*$ 
satisfies
$ A R^{-1} A^T \p^* = \chi$.

\begin{definition}[Certifying $\eps$-Electrical Flow Algorithm]
Let $\eps > 0$.
% \begin{enumerate}
%   \item
  A \emph{certifying $\eps$-electrical flow algorithm} is an algorithm that computes an
  $\eps$-electrical flow $\f$ and voltages $\p\in \QQ^n$ such that 
  \[
    \|\p - \p^*\|_{A R^{-1} A^T}^2
    \le \eps \|\p^*\|_{A R^{-1} A^T}^2,
  \]
  where $\p^*$ is an optimal dual solution.
  % \item 
  We define $\CEF(n,m,\eps)$ to be a bound on the running time of a certifying 
  $\eps$-electrical flow algorithm for directed graphs with $n$ nodes and $m$ arcs.
  % \end{enumerate}
\end{definition}
Kelner et al.~\cite{DBLP:conf/stoc/KelnerOSZ13} present a combinatorial 
$\eps$-electrical flow algorithm with expected approximation guarantee.
However, we transform their algorithm to one with an exact approximation guarantee 
and linear running time with high probability. 
Similarly to them, we compute a low-stretch spanning tree $T$ 
(w.r.t.~the resistances), which has tree condition number 
$\tau(T)=O(m\log n \log\log n)$ using the method of Abraham and 
Neiman~\cite{DBLP:conf/stoc/AbrahamN12} that runs in $O(m\log n \log \log n)$.
We then sample non-tree edges $a$ according to the same probability distribution 
$p_a:=\frac{1}{\tau(T)}\frac{r(C_a)}{r_a}$ and push flow along the cycle $C_a$ 
until the gap between primal and dual objective value becomes 
less than $\epsilon$. The running time of this approach is 
$O(m\log^2n \log (n/\eps) \log\log n) = \tilde O(m)$ for $\eps^{-1}=O(\poly(n))$
with high probability as we show in Theorem~\ref{runtimepseudocode}.
Note that it suffices for our purpose to mimic their \texttt{SimpleSolver}, 
which scales with $\log(n/\eps)$ instead of $\log(1/\eps)$ as their improved version 
does. We remark that, as in their solver, the flow updates should be
performed using a special tree data structure~\cite[Section~5]{DBLP:conf/stoc/KelnerOSZ13}, 
which allows updating the flow in $O(\log n)$. Moreover, $\gap$ should only be computed every $m$ iterations, which results in $O(1)$ 
amortized time per iteration for the update of $\gap$.
% We remark that the electrical flow $f$ output by the algorithm and all
% intermediate results are rational if the vectors $\chi$ and $r$ were.

\subsection{The Method}
Using any certifying $\eps$-electrical flow algorithm, 
we can compute an approximation of $d'$ by solving problem~\eqref{elflow} and obtain an $\eps$-electrical flow $f$. 
In the electrical flow problem the resistances $R$ are given by $X^{-2}$ and the current sources $\chi$ 
by $\bar A g'$. 
We compute a cycle $\xh = g'-X^{-1}f$ from the flow as well as a cut $\sh=\bar A^T \p$
from the voltages $\p$.
The idea is to push flow around the cycle $\xh$ in a primal step, whereas, in a dual step,
we modify the slacks along the cut $\sh$. 
In Ye's algorithm the decision whether to make a primal or dual step is made 
dependent on $\|d'\|_2$. In our setting, however, we do not know the exact projection 
$d'$ of $g'$. Nevertheless, we can show that the 2-norm of $z' = g' - \sh$ does not 
differ too much from $\|d'\|_2$, so deciding dependent on $\|z'\|_2^2$ is possible.
We note that another crucial difference between Ye's algorithm and our Combinatorial 
Potential Reduction Algorithm is that we normalize by $\max\{1, \|\xh\|_{\infty}\}$
in the primal step, where in Ye's algorithm the normalization is done with $\|\xh\|_2$,
which requires taking square roots and could thus yield irrational numbers.

We remark that our method works with any certifying $\eps$-electrical
flow algorithm. However, we merge
the version of the \texttt{SimpleSolver} of Kelner et 
al.~\cite{DBLP:conf/stoc/KelnerOSZ13} as described above 
in our pseudocode implementation of Algorithm~\ref{potredalg} to be more self-contained.

\begin{algorithm}[t]
\DontPrintSemicolon
\SetKwData{Left}{left}\SetKwData{This}{this}\SetKwData{Up}{up}
\SetKwFunction{Union}{Union}\SetKwFunction{FindCompress}{FindCompress}
\SetKwInOut{Input}{Input}\SetKwInOut{Output}{Output}
\Input{Feasible flow $x>0$ and feasible dual variables $y$ and $s>0$, parameter $\delta$}
%  obtained as described in Algorithm~\ref{balarcs}
\Output{Feasible flow $x>0$ and feasible dual variables $y$ and $s>0$ s.t.~$x^Ts<1$.}
\BlankLine
	\While{$x^Ts\ge 1$}{
		$g':=  \frac{q}{x^Ts} Xs - \ones $, $\chi:= \bar A g'$, $r:=X^{-2}\ones$\;
    /* \emph{$\eps$-electrical flow computation, similar to} 
    \texttt{SimpleSolver} \emph{in}~\cite{DBLP:conf/stoc/KelnerOSZ13} */\;
    $T:=$ low-stretch spanning tree w.r.t.~$r$, $\tau(T):=\sum_{a\in A\setminus T} \frac{r(C_a)}{r_a}$, 
    $p_a:= \frac{r(C_a)}{\tau(T) r_a} $\;
    $\f:=$ tree solution with $A\f=\chi$ for $T$,
    $\p:=$ tree induced voltages of $\f$ \;
    $\gap:=\f^TR\f - 2\p^T\chi + \p^T\bar A \bar A^T \p$\;
	  \Repeat{$\gap < \delta$}{
      Randomly sample $a\in A\setminus T$ with probability $p_a$\;
      Update $\f$ by pushing $\sum_{b\in C_a} r_b f_b/ r(C_a)$ flow through $C_a$ in the direction of $a$\;
      Occasionally compute tree-induced voltages $\p$ and $\gap$\;
    }
%     $\p:=$ tree induced voltages of $\f$ \;
    /* \emph{Move in primal or dual direction.} */\;
    Set $\xh:= g' - X^{-1}\f$, $ \sh = \bar A^T \p$ and $z'= g' - \sh$.\;
		\If{$\|z'\|_2^2 \ge 1/4$}{
		  Do a primal step, i.e.\ $x':= \ones 
      - \lambda \frac{\xh}{\max\{1,\|\xh\|_{\infty}\}}$, where $\lambda= 1/4$.\;}
		\Else{ Do a dual step, i.e.\ $s':= s' - \mu \sh$ and $y:= y + \mu \p$,
		where $\mu=\frac{\ones^Ts'}{q}$.\;}
  }
  \Return $x$ and $y,s$
\caption{Combinatorial Potential Reduction Algorithm\label{potredalg}}
\end{algorithm}

\subsection{Analysis}
It is not hard to see that that the primal and dual steps  in the algorithm 
are in fact feasible moves. 
% The proof is delegated to the appendix.
\begin{lemma}\label{feasible}
 The new iterates $\bar x = X(\ones-\lambda \frac{\xh}{\max\{1,\|\xh\|_{\infty}\}})$,
 $\bar y = y + \mu \pi$ and $\bar s = X^{-1}(s'-\mu \sh)$ are feasible.
\end{lemma}
\begin{proof}
  Clearly, $A X \xh = A X (g' - X^{-1}f)= \chi - Af = 0$.
  Note that $X \bar x'>0$ if and only if $\bar x' >0$. 
 It holds that
 \begin{align*}
  \bar x'_a 
  =   1 - \lambda \frac{\xh_a}{\max\{1,\|\xh\|_{\infty}\}}
  \ge 1 - \lambda \frac{\|\xh\|_{\infty}}{\max\{1,\|\xh\|_{\infty}\}}
  \ge 1 - \lambda = 3/4
  \quad\text{for every }
  a\in A.
 \end{align*}
For the dual variables, we have $ A^T (y + \mu \p) + X^{-1}(s'-\mu \sh)
= A^T y + \mu A^T \p + s- \mu X^{-1} \bar A^T\p = c$.
In addition, we obtain
\begin{align*}
 \bar s'_a = s'_a - \mu \sh_a
  = x_a s_a - \frac{x^T s}{q} (g'_a - z'_a )
  = \frac{x^Ts}{q}( 1 + z'_a )
  \ge \frac{1}{2q}
  > 0,
\end{align*}
since $|z_a'|\le \|z'\|_{\infty} \le \|z'\|_2 \le 1/2$ and $x^Ts\ge 1$. 
\end{proof}

The following lemma shows that the potential is reduced by a constant
amount in each step. 
% The proof may be found in the appendix as well.
We remark that although the proof for the dual step is essentially similar to
the proof for Ye's algorithm, the normalization with the $\infty$-norm requires 
non-trivial changes in the proof for the primal step.
\begin{lemma}
  \label{constred}
 If $\delta\le 1/8$ and $p^2\ge m \ge 4$, 
 the potential reduction is constant in each step.
 % if $\delta\le 1/8$, then 
 \[
  \Pot(\ones,s')-\Pot(\ones-\lambda \frac{\xh}{\max\{1,\|\xh\|_{\infty}\}},s') \ge 1/64 
  \quad\text{ and } \quad
  \Pot(\ones,s')-\Pot(\ones,s' - \mu \sh) \ge \frac{1}{12}.
 \]
\end{lemma}
\begin{proof} 
\begin{enumerate}
\item
	We first show the estimate for the primal step. 
  Let $v$ be any vector with $\|v\|_{\infty} \le 1$, then
  \begin{align}
  \label{vbound}
  \begin{split}
    \Pot(\ones,s')-&\Pot(\ones-\lambda v,s') 
    =    - q \ln \Big(1 - \lambda \frac{v^T s'}{\ones^T s'}\Big) 
          + \sum_{a \in A} \ln\big(1 - \lambda v_a\big)\\
    &\ge q \lambda \frac{v^T s'}{\ones^T s'} 
        - \lambda \sum_{a\in A} v_a 
        - \frac{\lambda^2}{2(1-\lambda)}\sum_{a\in A} v_a^2
    =    \lambda g'^T v 
        - \frac{\lambda^2\| v \|_2^2}{2(1-\lambda)},
  \end{split}
  \end{align}
  where the inequality follows because $\ln(1+\gamma)\ge \gamma - |\gamma|^2/(2(1 - |\gamma|))$ for
  any $\gamma\in (-1,1)$.
  The variable $\gap$ from Algorithm~\ref{potredalg} can be written as
  \begin{equation}
  \begin{aligned}%\begin{split}
  \label{gap}
  	\gap
  	:&= 
  	\f^TR\f - 2\p^T\chi + \p^T\bar A \bar A^T \p
  	=
  	\| g' - \xh\|_2^2 - 2 g'^T (g' - z') + \|g' - z'\|_2^2\\
  	% &=
  	% \| g' - \xh\|_2^2 - \| g'\|_2^2  + \| g' - (g' - z')\|_2^2
  	&= 
  	\| g' - \xh\|_2^2 - \| g'\|_2^2  + \| z' \|_2^2
  	=
  	-2g'^T\xh + \|\xh\|_2^2 + \|z'\|_2^2,
  % \end{split}
  \end{aligned}
  \end{equation}
  which for the primal step, where $\|z'\|_2^2\ge 1/4$, yields the estimate
  \begin{align}
  \label{gtd}
  	2 g'^T\xh 
  	= 
  	\|\xh\|_2^2 + \|z'\|_2^2 - \gap
  	\ge
  	\|\xh\|_2^2 + 1/8,
  	\quad \text{ since }
  	\gap < \delta\le 1/8.
  \end{align}
  % when we run Algorithm~\ref{potredalg} with $\delta\le 1/8$.
  \begin{description}
    \item[Case $\max\{1,\|\xh\|_{\infty}\} = 1$:]
      Then, from \eqref{vbound} with $v=\xh$ and \eqref{gtd} we obtain
      \begin{align*}
        \Pot(\ones,s')-\Pot(\ones-\lambda \xh,s') 
        &\ge
        \frac{\lambda( \|\xh\|_2^2 + 1/8 )}{2} - \frac{\lambda^2}{2(1-\lambda)}
              \|\xh\|_2^2\\
        &=
        \frac{1}{2}
        \Big[ \Big(\lambda - \frac{\lambda^2}{1 - \lambda}\Big)\|\xh\|_2^2 + \frac{\lambda}{8}\Big]
        \ge \frac{1}{64} \quad \text{ for } \lambda = 1/4.
      \end{align*}
  \item[Case $\max\{1,\|\xh\|_{\infty}\} = \|\xh\|_{\infty}$:]
    We use \eqref{vbound} with $v=\xh/ \|\xh\|_{\infty}$ and \eqref{gtd}. Then 
    we conclude
    \begin{align*}
      \Pot(\ones,s')-\Pot(\ones-\frac{\lambda \xh}{\|\xh\|_{\infty}},s')
      &\ge
      \frac{\lambda( \|\xh\|_2^2 + 1/8 )}
            {2 \|\xh\|_{\infty}} - \frac{\lambda^2}{2(1-\lambda)}
            \frac{\|\xh\|_2^2}{\|\xh\|_{\infty}^2}\\
      &\ge
      \frac{1}{2}
      \Big[\frac{ \|\xh\|_2^2 + 1/8 }
            {\|\xh\|_2^2} \lambda - \frac{\lambda^2}{1-\lambda}\Big]
            \frac{\|\xh\|_2^2}{\|\xh\|_{\infty}^2}
      \ge
      \frac{1}{12}  \hfill \text{ for }\hfill \lambda = \frac{1}{4}. 
    \end{align*}

  \end{description}
  \item
    For the dual step, observe that $\bar s' := s' - \mu \sh = \frac{\ones^Ts'}{q} (\ones + z')$. We obtain
\begin{align*}
 \Pot(\ones,s')-\Pot(\ones, \bar s')
 &=   - q \ln \Big(\frac{m + \ones^T z'}{q}\Big)   
      - \sum_{a\in A} \ln s'_a + \sum_{a\in A} \ln \bar s'_a \\
 &\ge -  q \ln \Big(\frac{m + \ones^T z'}{q}\Big) 
      - m\ln \Big( \frac{\ones^T s'}{m} \Big)
      + \sum_{a\in A} \ln \Big(\frac{\ones^T s'}{q} (1 + z_a')\Big) \\
 &= - p \ln \Big(1+\frac{\ones^T z' - p}{q}\Big)   
      - m \ln \Big(1 + \frac{\ones^Tz'}{m} \Big)
      + \sum_{a\in A} \ln (1 + z_a').
\end{align*}
In the dual step we have $\|z'\|_2^2\le 1/4$ and therefore $\|z'\|_{\infty}<1$.
Using this, we obtain
\begin{align*}
\Pot(\ones,s')-\Pot(\ones,\bar s')
&\ge \frac{p^2-p\ones^T z'}{q}
      - \ones^T z'
      + \ones^Tz' - \frac{\|z'\|_2^2}{2(1-\|z'\|_{\infty})} \\
&\ge \frac{p^2-p\sqrt{m}\|z'\|_2}{p+m} 
      - \frac{\|z'\|_2^{2}}{2(1-\|z'\|_2)} \\
&\ge \frac{p^2 - p^2/2}{p+p^2} - \frac{1/4}{2(1-1/2)} 
= \frac{p}{2(p+1)} - \frac{1}{4} \ge \frac{1}{12}
\end{align*}
using $p^2\ge m \ge 4$ and $\|z'\|_2<1/2$. \qedhere
 \end{enumerate}
 \end{proof}

We already remarked that $\Pot(x,s)<0$ implies $x^Ts<1$, hence
$\Pot(x,s)\ge 0$ holds throughout the algorithm.
With Lemma~\ref{constred}, the initial potential bounds the number of iterations.

\begin{theorem}
\label{runtimepseudocode}
  Given primal and dual interior points with potential $P_0$ as input, 
  Algorithm~\ref{potredalg} outputs interior primal and dual solutions $x$ and $y,s$ 
  with $x^Ts<1$ after $O(P_0)$ iterations.
  It can be implemented such that it terminates after 
  \[
    O( P_0 \cdot m \log^3(m) \log \log m))
  \]
  time with probability at least $1-\exp(- m \log^3 (m) \log \log m)$.
\end{theorem}

\begin{proof}
	\providecommand{\fopt}{{\f^{*}}}
	Kelner et al.\ give the following convergence result~\cite[Theorem 4.1]{DBLP:conf/stoc/KelnerOSZ13}
	\begin{align}
	\label{expectbound}
		\Exp[f_j^TRf_j -\fopt^TR\fopt] \le \Big(1-\frac{1}{\tau}\Big)^j (f_0^TRf_0 - \fopt^T R \fopt),
	\end{align}
	here $\fopt$ denotes an optimal electrical flow and $f_j$ the flow computed in the 
	$j$'th iteration. Let $\gap_j$ denote the value of $\gap$
	in the $j$'th iteration and let $X$ denote a random variable counting 
	the number of iterations of the Repeat-Until loop in 
	Algorithm~\ref{potredalg}. It follows that
	\begin{align*}
		\Pr\big[ X > i \big]
		% &=
		% \prod_{j=1}^i \Pr\big[ \gap_j \ge \delta \big]
		=
		\prod_{j=1}^i \Pr\big[ f_j^TR f_j - 2 \p_j^T\chi + \p_j^T \bar A \bar A^T \p_j \ge \delta \big] 
		\le 
		\prod_{j=1}^i \Pr\big[ f_j^TR f_j - \fopt^TR\fopt \ge \frac{\delta}{\tau} \big],
	\end{align*}
	since $f_j^T R f_j - \fopt^T R \fopt \ge \gap_j/\tau$, 
	see~\cite[Lem. 6.2]{DBLP:conf/stoc/KelnerOSZ13}. Using Markov's bound, 
	equation~\eqref{expectbound} and the 
	bound on the initial energy $f_0^T R f_0 \le \st(T) \fopt^T R \fopt$, 
	see~\cite[Lem. 6.1]{DBLP:conf/stoc/KelnerOSZ13}, yields
	\begin{align*}
		\Pr\big[ X > i \big]
		&\le 
		\prod_{j=1}^i \frac{\Exp[f_j^TR f_j - \fopt^TR\fopt ]}{\delta/\tau}
		\le
		% \Big(\frac{\tau}{\delta} 
		% (f_0^TRf_0 - \fopt^T R \fopt) \Big)  ^ i 
		% \prod_{j=1}^i \Big(1-\frac{1}{\tau}\Big)^j\\
		% &=
		\Big(\frac{\tau}{\delta} 
		(f_0^TRf_0 - \fopt^T R \fopt) \Big)  ^ i 
		\Big(1-\frac{1}{\tau}\Big)^{\frac{i(i+1)}{2}}\\
		&\le
		\Big(\frac{\tau}{\delta} 
		(\st(T) -1) \|g'-d'\|_2^2 \Big)^ i 
		\Big(1-\frac{1}{\tau}\Big)^{\frac{i(i+1)}{2}}
		% \le
		% \Big(C_1 m^5 \log(m) \log \log m 
		% \Big)  ^ i 
		% \Big(1-\frac{1}{C_2 m}\Big)^{C_3i^2},
		= \exp(-m \log^3(m) \log\log m)
	\end{align*}
	with $i=O(m \log^2(m) \log \log m)$, and the guarantee on the low-stretch spanning
	tree of Abraham and Neiman~\cite{DBLP:conf/stoc/AbrahamN12}, which yields
	$\tau = O(\st(T)) = O(m\log(m)\log\log m)$. Hence, the 
	number of times the Repeat-Until loop is executed during one of the $O(P_0)$
	iterations is bounded by $O(m \log^2(m) \log \log m)$ with exponentially high
	probability. We remark that the updates of the flow $f$ should not be done in the 
	naive way but using a simple data structure
	exactly as it is also described in Kelner et al~\cite{DBLP:conf/stoc/KelnerOSZ13} for 
	their \texttt{SimpleSolver}. One iteration takes $O(\log n)$ time then. 
	We can compute $\gap$ in every $m$'th iteration in $O(m)$
	time, which yields that we make at most $m$ steps to much and need amortized constant
	run-time for the update of $\gap$ in each iteration. This yields the bound.
\end{proof}

% The theorem follows by an application of Markov's inequality to the probability
% that the gap is larger than $\delta$ after one iteration of the Repeat-Until loop.
We remark that we can also keep running the algorithm until $x^Ts<c$
for any $c\in\RRnneg$ without affecting the running time.
In addition, we get the following more general result. To prove it, it
remains to show that a $1/(16q^2)$-electrical flow fulfills $\gap\le 1/8$.
% This proof may be found in the appendix as well.
\begin{theorem}
\label{runtimegeneral}
  Given primal and dual interior points with a potential of $P_0$ as input, 
  there is a combinatorial algorithm that outputs interior primal and dual 
  solutions $x$ and $y,s$ with $x^Ts<1$ and needs
  \[
    O( P_0 \cdot \CEF(n,m,1/(16q^2)))
  \]
  time, where $q=m+\min\{k\in \ZZ : k^2\ge m\}$ and $\CEF(n,m,\eps)$ 
  is the running time of a certifying $\eps$-electrical flow algorithm.
\end{theorem}
\begin{proof}
	It remains to show that a $1/(16q^2)$-electrical flow fulfills $\gap\le 1/8$.
  The approximation guarantee from the certifying $\eps$-electrical flow 
  algorithm for the primal and dual solution yield
  \begin{align}
  \label{primgar}
  \begin{split}
    \|g'-\xh \|_2^2 \le (1+\eps)\|g'-d'\|_2^2
    \hspace{3mm} \text{and} \hspace{3mm} 
  	\| \p - \p^*\|_{\bar A \bar A^T}^2
    \le \eps \|\p^*\|_{\bar A \bar A^T}^2,
  \end{split}
  \end{align}
  the second guarantee equivalently writes as 
  \begin{align*}
  	\eps \|g'-d'\|_2^2
  	\ge 
  	\|z'-d'\|_2^2 
  	=\|z'\|_2^2 -2d'^T(g'-\bar A^T\p) + \|d'\|_2^2
  	=\|z'\|_2^2 - \|d'\|_2^2.
  \end{align*}
  Together with~\eqref{gap} and~\eqref{primgar}, we obtain
  \begin{align*}
  	\gap
  	&=
  	\|g'-\xh \|_2^2 - \|g'\|_2^2 + \|z'\|_2^2
  	\le
  	(1+\eps) \| g'- d' \|_2^2 - \| g' \|_2^2 + \eps \| g'- d' \|_2^2 
  		+  \|d'\|_2^2\\
  	&\le
  	2\eps \| g'- d' \|_2^2 \le 2\eps \| g' \|_2^2 \le 2 \eps q^2 \le \frac{1}{8}.\qedhere
  \end{align*}
\end{proof}

\section{Initialization}\label{sec:initial}
In this section, we describe how to find initial points with $P_0=\tilde O(\sqrt m)$
that we can use to initialize Algorithm~\ref{potredalg}.
We assume w.l.o.g. that the given min-cost flow 
instance is 
% feasible and 
finite, that the capacities are finite and that the costs
are non-negative. In order to be self-contained, we also justify these assumptions. 

We first describe how one recognizes unbounded instances. 
Consider the graph $G_{\infty}=(V,A_{\infty})$, where 
$A_{\infty}:=\{a\in A_0 : u_a =\infty\}$ denotes the set of arcs with infinite 
capacity.
By running a shortest path algorithm for graphs with possibly negative arc length, 
as for example the one presented by Goldberg in~\cite{DBLP:conf/soda/Goldberg93}, 
we can detect whether $G_{\infty}$ contains a negative cycle in $O(\sqrt n m \log C)$ 
time. If there is such a uncapacitated negative cycle, the problem is unbounded and the 
solution is $-\infty$, otherwise the solution is finite.
Now, since we know that the problem is finite, provided that it is feasible as well, there will always be an optimal basic solution. Hence, the maximum flow on any arc in this solution will be bounded by 
$\|b\|_1/2$.
Hence, we set the capacity of every uncapacitated arc to $u_a = \|b\|_1/2$.
There is also a well-known technique to remove the negative costs:
Saturate the arcs with negative cost and consider the residual network,
this gives an equivalent problem with $c\ge 0$. Note that the increase in 
$\|b\|_1$ due to this construction is only polynomial.
We remark that we do not need to check feasibility,
since the crossover procedure presented above enables us to recognize 
infeasibility. This is described at the end of this section.

\subsection{Removing Capacity Constraints}
Using a standard reduction, we modify the network in order to get rid of the upper 
bound constraints $x\le u$. We briefly review the construction since we will later
extend it to obtain the auxiliary network flow problem with interior primal and dual 
points of low potential. Let $G_0=(V_0,A_0)$ denote the original input graph. For an 
edge $a=(v,w)\in A_0$, we proceed as follows, see from left to middle in 
Figure~\ref{fig:remcapfig}: Remove $a$, insert a node $vw$, insert arcs 
$\acute{a}=(v,vw)$ and $\grave{a}=(w,vw)$ with $c_{\acute{a}} = c_a$ and 
$c_{\grave{a}} =0$, respectively.\footnote{
	The accents reflect the direction in which the
	arc is drawn in Figure~\ref{fig:remcapfig}.
} 
Moreover, set $b_{vw} = u_a$ and subtract $u_a$ from $b_w$. 

\begin{figure}[ht!]
	\begin{center}
		\resizebox{!}{3.4cm}{
			\begin{tikzpicture}[shorten >=1pt,node distance=2cm,>=stealth',initial/.style={}]
\begin{scope}
\tikzstyle{every state}=[draw=blue!50,very thick,fill=blue!20]
  \node[state, color = white]			(vw) 			{\large{$vw$}};
  \node[state]          	(w) [above left =of vw]	{\large{$w$}};
  \node[state]          	(v) [below left =of vw]	{\large{$v$}};
\tikzset{mystyle/.style={orange}} 
% %\path (v)     edge [mystyle]    node   {$3$} (w);

\path[->]	(v) edge 		node[left] {\large{$(u_a,c_a)$}} (w);
\path[->]		(w) edge[mystyle] 	node[above] {\large{$b_w$}}  ++ (-1.5,0);
\path[->]		(v) edge[mystyle] 	node[above] {\large{$b_v$}}  ++ (-1.5,0);
%\path[<-]		(vw) edge[mystyle,color=white] node[above] {}  ++ (3,0);
\end{scope}

\begin{scope}[xshift=4cm]
\tikzstyle{every state}=[draw=blue!50,very thick,fill=blue!20]
  \node[state]			(vw) 			{\large{$vw$}};
  \node[state,color=white]	(a) [left =of vw]	{};
  \node[state]          	(w) [above left =of vw]	{\large{$w$}};
  \node[state]          	(v) [below left =of vw]	{\large{$v$}};
\tikzset{mystyle/.style={orange}} 
%\path (v)     edge [mystyle]    node   {$3$} (w);
\path[->]	(v) edge 		node[below right] {\large{$(\infty,c_a)$}} (vw);
\path[->]	(w) edge 		node[above right] {\large{$(\infty,0)$}} (vw);
\path[->]		(v) edge[mystyle] 	node[above] {\large{$b_v$}}  ++ (-1.9,0);
\path[->]		(w) edge[mystyle] 	node[above] {}  ++ (-1.9,0);
\node[color=orange] (d) at (-3.17,2.4) {\large{$b_w - u_a$}};
%\path[->]		(vw) edge[mystyle,color=white] 	node[above] {$u_{a}$}  ++ (3,0);
\path[->]		(vw) edge[mystyle] 	node[above] {\large{$u_{a}$}}  ++ (1.7,0);
%\path[->]		(a) edge[mystyle,color=white] node[above] {}  ++ (-2.5,0);
\end{scope}

\begin{scope}[xshift=10 cm]
\tikzstyle{every state}=[draw=blue!50,very thick,fill=blue!20]
  \node[state]			(vw) 			{\large{$vw$}};
  \node[state]          	(w) [above left =of vw]	{\large{$w$}};
  \node[state]          	(v) [below left =of vw]	{\large{$v$}};
\tikzset{mystyle/.style={orange}} 
%\path (v)     edge [mystyle]    node   {$3$} (w);
\path[->]	(v) edge 		node[below right] {\large{$c_{a}$}} (vw);
\path[->]	(w) edge 		node[above right] {\large{$0$}} (vw);
\path[->]	(v) edge 		node[left] {\large{$c_{\hat a}$}} (w);
\path[->]		(w) edge[mystyle] 	node[above] {\large{$b_w$}}  ++ (-1.5,0);
\path[->]		(v) edge[mystyle] 	node[above] {\large{$b_v$}}  ++ (-1.5,0);
%\path[<-]		(vw) edge[mystyle,color=white] node[above] {}  ++ (3,0);
\path[->]		(vw) edge[mystyle] 	node[above] {\large{$b_{vw}$}}  ++ (1.5,0);
\end{scope}
\end{tikzpicture}
		}
	\end{center}
	\caption{
	The transition from the left to the middle, which is done for each arc, removes the 
	capacity constraint. From the middle to the right: In order to balance the $x_as_a$,
	we introduce the arc $\hat a = (v,w)$ with high cost $c_{\hat a}$ and reroute flow 
	along it. The direction of $\hat a$ depends on a tree solution $z$ in $G_0$. It is flipped if $z_a \le u_a / 2$.}
\label{fig:remcapfig}
\end{figure}
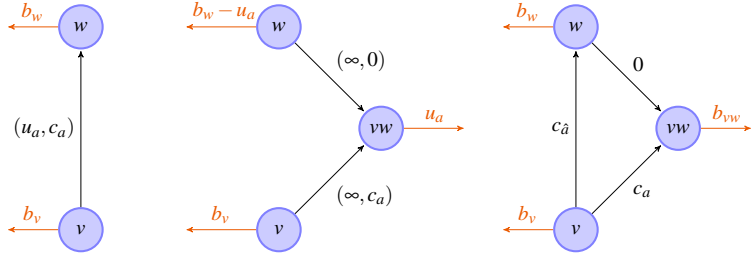

% We remark that we can even assume $|b'(v)|>0$ for all $v\in V_1$: 
% If $b'(v)=0$, then for any feasible flow $x$ it holds that $x_a=u_a$
% for all $a\in \din(v)$. So we remove all arcs $a\in \din(v)$ from 
% the network and replace $b'(v)$ by $b'(v) - u(\din(v))$.

\subsection{Finding the Initial Flow}
Recall the equivalent form of the potential function in \eqref{equipot}.
It illustrates that the potential becomes small if the ratio between the 
arithmetic and the geometric mean does. This in turn is the case if the variance of 
the $x_a s_a$ is low over all $a\in A$. This observation is crucial for 
our Algorithm~\ref{balarcs} that finds an initial flow with low potential.
Our aim is to balance the flows on $(v,vw)$ and $(w,vw)$, by introducing
the arc $(v,w)$ or $(w,v)$, see Figure~\ref{fig:remcapfig} from middle to right.
Since we perform the two transitions of Figure~\ref{fig:remcapfig} together, we will 
refer to the resulting graph as $G_1=(V_1,A_1)$ with $|V_1|=n_1$ and $|A_1|=m_1=3m$.
For the sake of presentation, we assume w.l.o.g.~that 
all capacities $u_a$ of the original graph $G$ were odd
%For $G_1$ this yields that $b_{vw}$ is odd for all $vw\in R$, 
such that $z_a - u_a/2 \ne 0$ for all integers $z_a$.
\footnote{This is justified
by the following argument: If the capacity of an arc is even, then we add a parallel 
arc of capacity 1 and reduce the capacity of the original arc by 1.}
\begin{algorithm}[ht]
\DontPrintSemicolon
\SetKwData{Left}{left}\SetKwData{This}{this}\SetKwData{Up}{up}
\SetKwFunction{Union}{Union}\SetKwFunction{FindCompress}{FindCompress}
\SetKwInOut{Input}{Input}\SetKwInOut{Output}{Output}
\Input{$G_0=(V_0,A_0)$, parameter $t$.}
\Output{Graph $G_1$, primal and dual solutions $x$, and $y,s$, such that $x_a s_a\in [t,t+CU/2]$.}
\BlankLine
  Compute a tree solution in $G_0$ and obtain an integral (not necessarily feasible) 
  flow $z$.\;
  \For{every arc $a\in A_0$}{
    Insert node $vw$, arcs $\vvw a=(v,vw)$, $\wvw a=(w,vw)$ with $c_{\vvw a}=c_a$, 
    $c_{\wvw a} = 0$, set $x_{\vvw a} = x_{\wvw a} = u_a/2$ \;
    \If{$z_a >u_a/2$} {
      Replace $a$ by $\hat{a}=(v,w)$
      }
    \Else{
      Replace $a$ by $\hat{a}=(w,v)$
    }
		$c_{\hat{a}}= \big\lceil t/|z_a - u_a/2| \big\rceil$, 
    $x_{\hat a} := |z_a -u_a/2|$, 
    $y_{vw}:= - 2t/u_a$ and $y_v,y_w:=0$\;
  }
  \Return the resulting graph $G_1=(V_1,A_1)$ and $x,y$ with corresponding slacks $s$.
\caption{Balance Arcs\label{balarcs}}
\end{algorithm}

\begin{theorem}
  \label{initpot}
  Let $G_1=(V_1,A_1)$, $x,y,s$ be output by Algorithm~\ref{balarcs} and $\Gamma:= \max\{C,U,\|b\|_1/2\}$.
  \begin{enumerate}
    \item It holds that
    $x_a s_a \in [t,t+ \Gamma^2]$ for all 
    $a\in A_1$.
    \item Setting $t=m\Gamma^3$ and $p=\min\{k\in\ZZ: k^2 \ge m_1\}$ yields
    $
      \Pot(x,s) = O(\sqrt m \log(n\gamma)).
    $
  \end{enumerate}
\end{theorem}
\begin{proof}
\begin{enumerate}
  \item
  Let $a\in A_0$ be any arc in $G_0$. We have
   $
     x_{\vvw a} = 
     x_{\wvw a} = 
     u_a/2.
   $
  It holds that, 
  \begin{align*}
    &x_{\vvw a} s_{\vvw a} = 
    \frac{u_a}{2} \left(c_{\vvw a} + \frac{2t}{u_a}\right)
    = t + \frac{ u_a c_a}{2} \le t + \Gamma^2, 
  \quad 
  % \text{and} \quad
    x_{\wvw a} s_{\wvw a} =
    \frac{u_a}{2} \left(c_{\wvw a} + \frac{2t}{u_a} \right)
    =t
  \quad \text{and}\quad\\
  %  \]
  % For $\hat a$, we have
  % \begin{align*}
    &x_{\hat a} s_{\hat a}
   % = x_{\hat a} c_{\hat a}
   \ge \Big|z_a - \frac{u_a}{2}\Big| %&
   \frac{t}{|z_a - \frac{u_a}{2}|}  
    = t \quad \text{and} \quad %\\
   %&
   x_{\hat a} s_{\hat a}
   \le \Big|z_a - \frac{u_a}{2}\Big| \left(\frac{t}{|z_a - \frac{u_a}{2}|} + 1 \right)  
   %\le t + \frac{u_a}{2} 
   \le t + \Gamma.
   \end{align*}
   % In summary, for all arcs $a\in A_1$, it holds that $x_a s_a \in [t,t + CU/2]$.

\item
  We consider the potential function with $q=m_1+p$, we will fix $p$ below.
  We have
  \begin{align*}
   \Pot(x,s) :&= q\ln ( x^Ts ) -\sum_{a\in A_1} \ln (x_a s_a) - m_1 \ln m_1
   \le q \ln \big( m_1  t + 2m\Gamma^2 \big)
   - m_1 \ln (m_1 t)\\
   &\le q \ln \Big(1 + \frac{m_1 \Gamma^2}{m_1 t}\Big) + p \ln (m_1 t)
   % \le q \ln \Big( 1 + \frac{1}{m_1} \Big)
   % +  p \ln (m_1^2 CU)%\\
   %&
   \le \frac{m_1+p}{m\Gamma} + p\ln (m_1^2 \Gamma^3),
   \text{ since }
   t=m\Gamma^3.
  \end{align*}
  For $p=\min\{z\in \ZZ : z^2\ge m_1 \}$, we get
  $ \Pot(x,s) = O(\sqrt m \log n\Gamma) = O (\sqrt{m} \log n \gamma)$.\qedhere
  \end{enumerate}
  \end{proof}
Algorithm~\ref{balarcs} can be implemented in $O(m)$ time.
We remark that, due to the high costs of the arcs in $A_1\setminus A_0$, there will 
never be flow on them in an optimal solution.
In particular, these arcs are more expensive than any path in the original
network because $c_{\hat a}=\lceil t/|z_a-u_a/2| \rceil \ge mCU$. Therefore, 
the optimum of the problem is not changed by the introduction of the arcs 
$\hat a$. Note that the resulting network is always feasible.
This is why we can assume feasibility in Section~\ref{sec:crossover}.

\section{Summary}
We first run Algorithm~\ref{balarcs} on the input graph
$G_0$ to construct the auxiliary network $G_1$.
We then initialize Algorithm~\ref{potredalg} with the obtained interior points. If 
$\lceil b^Ty^0\rceil >mCU$, the problem in $G_0$ was infeasible, since any solution
in $G_0$ is bounded by $mCU$. Otherwise, we apply Algorithm~\ref{crossoveralg} and obtain 
optimal integral potentials $y$ in $G_1$. Let $H_1$ be the admissible network, 
i.e.\ the graph $G_1$ with all arcs with dual slack $0$. Consider $H_0$, the graph
resulting by removing all arcs $\hat a$ from $H_1$ that were introduced by Algorithm~\ref{balarcs}.
By a max-flow computation 
we compute a feasible solution $x$ in $H_0$, which is optimal in $G_0$ by complementary 
slackness.
If $H_0$ is however infeasible, there is a set $S$ with $b(S)\le -1$ and 
$\dout_{H_0}(S)=\emptyset$~\cite[Corollary 11.2h]{schrijver}. Since $y$ is optimal 
in $G_1$, there is an arc $\hat a\in \dout_{G_1}(S)$ with
$s_{\hat a}=0$, thus $\hat a\in \dout_{H_1}(S)$. It follows that there is always a 
feasible and integral solution $z$ in $H_1$ with $z_{\hat a} \ge 1$ that is optimal 
in $G_1$. With $c_{\hat a} \ge mCU$, we conclude that the cost of $z$ is larger than 
$mCU$, which contradicts $\lceil b^Ty^0\rceil \le mCU$.
Since the max-flow computation requires $O(m^{3/2}\log(n^2/m)\log U )$
if it is carried out with the algorithm of Goldberg and 
Rao~\cite{DBLP:conf/focs/GoldbergR97}, this concludes the proof of 
Theorem~\ref{main theorem}.

\bibliography{cites}
\bibliographystyle{ieeetr}

\end{document}